\documentclass[conference]{IEEEtran}
\ifCLASSINFOpdf
\else
\fi
\usepackage{times}
\usepackage{microtype}
\usepackage{amsmath,amssymb,amsfonts}

\usepackage{multirow}
\usepackage{enumitem}

\usepackage{cite}
\usepackage{makecell}
\usepackage{amsmath,amssymb,amsfonts}
\usepackage{graphicx}
\usepackage{textcomp}
\usepackage{xcolor}
\usepackage{cite}
\usepackage{amsmath}
\usepackage{amssymb}
\usepackage{amsthm}
\usepackage{algorithm}
\usepackage{algpseudocode}
\usepackage{tikz}
\usepackage{adjustbox}
\usepackage{pgfplots}
\pgfplotsset{compat=1.18}
\usepackage{pgf-umlsd}
\usepackage{subfigure}
\usepackage{subcaption}
\def\BibTeX{{\rm B\kern-.05em{\sc i\kern-.025em b}\kern-.08em
    T\kern-.1667em\lower.7ex\hbox{E}\kern-.125emX}}
\usepackage{amsthm}
\newtheorem{theorem}{Theorem}[section]
\newtheorem{lemma}[theorem]{Lemma}

\usepackage{booktabs}
   \usepackage{medmath} 
\usepackage{tikz}
\usetikzlibrary{arrows.meta, positioning, shapes.geometric, calc}
\algnewcommand\Input{\item[\textbf{Input:}]}
\algnewcommand\Output{\item[\textbf{Output:}]}
\usepackage{url}
\usepackage{microtype}          
\usepackage{enumitem}           
\usepackage{titlesec}           
\usepackage{setspace}           

\titlespacing*{\section}
{0pt}{1.1ex plus 0.3ex minus 0.2ex}{0.8ex}
\titlespacing*{\subsection}
{0pt}{0.9ex plus 0.2ex minus 0.2ex}{0.6ex}
\titlespacing*{\subsubsection}
{0pt}{0.7ex plus 0.2ex minus 0.2ex}{0.4ex}
\usepackage[font=small]{caption}

\setlist[itemize]{noitemsep, topsep=2pt}
\setlist[enumerate]{noitemsep, topsep=2pt}

\begin{document}
%
\title{Adaptive Quantum-Safe Cryptography for 6G Vehicular Networks via Context-Aware Optimization}

\author{\IEEEauthorblockN{Poushali Sengupta}
	\IEEEauthorblockA{University of Oslo\\
		poushals@ifi.uio.no}
	\and
	\IEEEauthorblockN{Mayank Raikwar}
	\IEEEauthorblockA{University of Oslo\\
		mayankr@ifi.uio.no}
	\and
	\IEEEauthorblockN{Sabita Maharjan\\}
	\IEEEauthorblockA{University of Oslo\\
		sabita@ifi.uio.no}
    \and
	\IEEEauthorblockN{Frank Eliassen\\ }
	\IEEEauthorblockA{University of Oslo\\
		frank@ifi.uio.no}
    \and
	\IEEEauthorblockN{Yan Zhang\\}
	\IEEEauthorblockA{University of Oslo\\
		yanzhang@ifi.uio.no}
        }
	

%


\IEEEoverridecommandlockouts
\makeatletter\def\@IEEEpubidpullup{6.5\baselineskip}\makeatother
\IEEEpubid{\parbox{\columnwidth}{
		Workshop on Security and Privacy of Next-Generation Networks \\ (FutureG) 2026 \\
		23 February 2026, San Diego, CA, USA \\
		ISBN 979-8-9919276-9-7\\ 
		https://dx.doi.org/10.14722/futureg.2026.23xxx \\   
		www.ndss-symposium.org
}
\hspace{\columnsep}\makebox[\columnwidth]{}}

\maketitle

\begin{abstract}
Powerful quantum computers in the future may be able to break the security used for communication between vehicles and other devices (Vehicle-to-Everything, or V2X). New security methods called post-quantum cryptography can help protect these systems, but they often require more computing power and can slow down communication, posing a challenge for fast 6G vehicle networks. In this paper, we propose an adaptive post-quantum cryptography (PQC) framework that predicts short-term mobility and channel variations and dynamically selects suitable lattice-, code-, or hash-based PQC configurations using a predictive multi-objective evolutionary algorithm (APMOEA) to meet vehicular latency and security constraints.However, frequent cryptographic reconfiguration in dynamic vehicular environments introduces new attack surfaces during algorithm transitions. A secure monotonic-upgrade protocol prevents downgrade, replay, and desynchronization attacks during transitions. Theoretical results show decision stability under bounded prediction error, latency boundedness under mobility drift, and correctness under small forecast noise. These results demonstrate a practical path toward quantum-safe cryptography in future 6G vehicular networks. Through extensive experiments based on realistic mobility (LuST), weather (ERA5), and NR-V2X channel traces, we show that the proposed framework reduces end-to-end latency by up to 27\%, lowers communication overhead by up to 65\%, and effectively stabilizes cryptographic switching behavior using reinforcement learning. Moreover, under the evaluated adversarial scenarios, the monotonic-upgrade protocol successfully prevents downgrade, replay, and desynchronization attacks.
\end{abstract}


%
\IEEEpeerreviewmaketitle

\section{Introduction}

As large-scale quantum computers become more prevalent, they pose a serious threat to the security of modern vehicle communication systems. Current security measures for Vehicle-to-Everything (V2X) communication mostly rely on techniques like elliptic-curve and RSA-based public-key infrastructures. Unfortunately, these methods are not effective enough, such as those using Shor's algorithm~\cite{shor1997algorithm}. For critical applications in future 6G vehicle networks, where safety and reliability are essential, switching to quantum-safe security is necessary.
\par  Standardized PQC algorithms include lattice-based schemes (e.g., Kyber and Dilithium)~\cite{nistr3report,dilithium2020spec,kyber2020spec}, code-based schemes (e.g., Classic McEliece)~\cite{mceliece1978,sendrier2011}, and hash-based signatures (e.g., SPHINCS+)~\cite{sphincs2019}. While PQC algorithm hold great potential, directly deploying PQC in vehicular systems remains challenging. Many PQC algorithms require larger keys, more processing power, and increased computation time resulted in communication delay~\cite{paul2022pqcAutomotive,pan2023pqcVANET}, which can interfere with real-time requirements such as collision avoidance and cooperative vehicle communication. Additionally, traditional models for implementing PQC often overlook the dynamic nature of vehicular contexts~\cite{ha2020v2xSurvey,liu2021v2x}. Vehicles face changing conditions, including varying mobility, weather-related impacts on communication quality, and fluctuating computing resources. Depending on these conditions, the demands can range from urgent safety notifications to routine data transfers. A static post-quantum cryptographic configuration cannot simultaneously satisfy all vehicular operating conditions, as overly conservative choices introduce unnecessary delay, while lightweight configurations may offer insufficient security under challenging or adversarial scenarios. Vehicular links exhibit high Doppler spread, frequent beam misalignment, and rapid channel variations~\cite{3gpp38885,ali2020v2x}, leading to transient cost differences between PQC families. Code-based schemes excel during deep fades due to robustness~\cite{mceliece_fading_2020}, while lattice-based schemes perform better in moderate-SNR, low-latency scenarios~\cite{jao2022pqcperf,pan2023pqc}. Because these fluctuations occur within sub-100 ms~\cite{khattak2019fastv2x}, static PQC assignments can cause latency or security overhead, underscoring the need for a predictive, context-aware cryptographic layer that optimizes PQC choices to maintain URLLC compliance~\cite{adaptivecrypto2020}.

To tackle these issues, we propose a Context-Aware Adaptive PQC (CAAP) Framework tailored for 6G vehicular networks. Unlike conventional, fixed PQC deployments, CAAP  dynamically selects the best cryptographic algorithm based on current conditions using real-time sensing and predictive analysis. The framework consists of four main components: (1) a \textit{Context-Sensing Pipeline} that collects vehicle speed, communication quality, weather conditions, and message urgency into a unified context vector; (2) a \textit{Short-Term Predictor} that anticipates context changes in 100–200\,ms windows; (3) an \textit{Adaptive Predictive Multi-Objective Evolutionary Algorithm (APMOEA)} that balances latency, compute cost, communication overhead, and quantum-resilience by selecting among lattice-based, code-based, or hash-based signatures; and (4) a \textit{Secure Transition Protocol} that prevents downgrade, replay, and context-manipulation attacks through authenticated version-monotonic negotiation, following principles similar to TLS~1.3 and QUIC~\cite{rescorla2018tls13,iyengar2021quicSecurity}.

\par At the core, APMOEA uses reinforcement learning to continually adapt and reduce processing and communication delays while maximizing security. Learning from real-time and past performance, it ensures effective PQC decisions across diverse vehicular environments. We rigorously analyze this adaptive layer, deriving stability conditions for context-aware decisions, bounded latency under mobility drift, and correctness of PQC choice under small prediction error. Through extensive testing under realistic conditions, including LuST mobility traces, ERA5 weather data, 3GPP-compliant NR-V2X channel models, and automotive compute models, our adaptive framework demonstrates a 27\% latency reduction, full downgrade-attack resistance, and improved robustness over NSGA-II and RL-only baselines. Main contributions in this paper are as follows:
\begin{itemize}[leftmargin=0.2cm, itemsep=1pt, topsep=1pt]
    \item \textbf{Adaptive planning for PQC selection under vehicular constraints.}  
    We propose APMOEA, a predictive multi-objective planning mechanism that
    dynamically selects PQC configuration profiles under stringent URLLC
    latency, computation, communication, and security constraints.

    \item \textbf{Secure execution via monotonic PQC transitions.}  
    We design a lightweight, authenticated transition protocol that enforces
    monotonic PQC upgrades and prevents downgrade, replay, and desynchronization
    attacks during reconfiguration.

    \item \textbf{Formal guarantees for stability and bounded latency.}  
    We provide theoretical analysis establishing decision stability under bounded
    prediction error, correctness relative to an oracle selector, and bounded
    end-to-end latency under realistic vehicular mobility and channel drift.

    \item \textbf{Claim-driven evaluation on realistic V2X traces.}  
    We conduct experiments using LuST mobility, ERA5 weather data,
    NR-V2X channel models, and automotive-grade compute profiles, demonstrating
    latency reductions of up to 27\% and robust behavior under adversarial
    context manipulation.
\end{itemize}
We emphasize that this work does not introduce new cryptographic primitives.
Instead, this work primarily addresses the problem of \emph{adaptive selection and secure transition of post-quantum cryptographic (PQC) mechanisms under stringent vehicular ultra-reliable low-latency communication (URLLC) constraints}. While the system is designed to be compatible with emerging 6G-ready V2X environments, the central contribution lies in \emph{context-aware cryptographic adaptation} rather than detailed physical-layer modeling of future 6G networks. Multi-objective optimization and reinforcement learning are employed as enabling mechanisms to realize this adaptive security objective efficiently and stably. As fully standardized 6G vehicular communication stacks are not yet available, we adopt a \emph{6G-ready} perspective. Specifically, we model vehicular communications that reflect anticipated 6G characteristics, such as URLLC, high mobility, and dense connectivity, while remaining compatible with current NR-V2X abstractions.

\textbf{Scope of PQC deployment:} CAAP applies post-quantum cryptography at the application layer to protect V2V data communications transmitted over established 5G connections. The PQC schemes encrypt and authenticate safety-critical and cooperative awareness messages exchanged between vehicles. This is distinct from and complementary to the network-layer security provided by 5G-AKA and EAP-TLS, which handle initial authentication and secure connection establishment between vehicles and the 5G network infrastructure.

\section{Background, Related Work, and Threat Model}
This section covers the background, prior work, and adversarial assumptions driving the need for adaptive post-quantum cryptography in V2X systems.
\subsection{Background and Related Work}
Next-generation vehicular network services, which will be developed on beyond 5G and 6G technologies demand ultra-reliable low-latency communication (URLLC) with latency requirements between 5 and 20 milliseconds, high reliability, and the ability to operate effectively under rapidly changing mobility and channel conditions \cite{3gpp_22_886, liu20226gv2x}. Traditional security methods based on Elliptic Curve Cryptography (ECC) and RSA are vulnerable to quantum computing threats, making the adoption of PQC crucial for future implementations \cite{shor1997polynomial, nist_pqc_2022}. However, when deployed on today’s classical vehicular hardware, post-quantum cryptographic schemes generally incur larger key sizes, higher verification costs, and increased computational overhead compared to traditional public-key cryptography, all of which can fluctuate significantly with factors such as SNR, packet error rate (PER), CPU load, and changes in mobility \cite{chen2021nist, jao2022pqcperf, ali2020v2x}. As a result, static PQC configurations often violate URLLC latency constraints under strict operating conditions or introduce avoidable computational and communication overhead during low-risk scenarios. 
\par Most recent proposals for PQC in V2X systems mainly evaluate algorithms such as Kyber, Dilithium, and McEliece under fixed conditions~\cite{pan2023pqc, chen2022vanet}. These evaluations do not account for the effects of rapid vehicle movement, the fading processes described by 3GPP, or variations in available computational resources. Additional benchmarking studies~\cite{chen2021nist, almeida2024pqc} has further demonstrated that no single PQC method outperforms others across all operational scenarios. This highlights the necessity for a PQC layer that can adapt based on predictions and contextual awareness. While adaptive and resource-aware cryptographic strategies for the Internet of Things (IoT) and embedded systems \cite{sun2018adaptive, rahman2020context} adjust classical key strengths according to load or energy availability, these approaches do not handle significant variability in PQC parameters nor do they protect against downgrade, rollback, or context-forcing attacks that may occur during PQC transitions. Protocols like TLS 1.3 and QUIC \cite{rfc8446, iyengar2020quic} provide mechanisms to resist downgrade attacks, but assume stable client-server environments, unlike the dynamic conditions in V2X scenarios, which involve unpredictable latency. Prior research on detecting downgrades \cite{rescorla2018tls13} and negotiating QUIC versions \cite{quic2021} has highlighted the challenges of preventing rollback in hostile network environments, but has not integrated real-time context aware adaptation. To the best of our knowledge, CAAP is the first to merge downgrade-resistant PQC transitions with a reinforcement-learning-guided predictive optimizer, specifically tailored for the demands of 6G V2X operations. This work does not introduce new post-quantum cryptographic primitives; instead, it focuses on the adaptive orchestration of standardized PQC mechanisms (e.g., Kyber, Dilithium, SPHINCS+) in response to rapidly changing vehicular contexts.
\subsection{Threat Model}
We consider a powerful adversary with access to V2X wireless communications who can inject, replay, delay, or modify packets~\cite{rfc8446,iyengar2020quic}. The attacker may forge contextual features such as SNR or PER to influence PQC selection~\cite{rahman2020context}, replay outdated version counters to trigger downgrade attacks~\cite{quic2021}, or manipulate message delivery to induce synchronization failures~\cite{rescorla2018tls13}. The adversary’s objectives include weakening PQC strength, increasing latency, destabilizing transitions, or steering the system toward slower or less secure configurations. In practice, post-quantum security relies on coordinated cryptographic mechanisms rather than a single primitive. Each candidate action therefore represents a \emph{PQC configuration profile} comprising (i) a key encapsulation mechanism (KEM) for session key establishment and (ii) a digital signature scheme for authentication and integrity. For example, configurations may combine Kyber for key exchange with Dilithium or SPHINCS$+$ for authentication. We assume standardized, NIST-recommended PQC primitives and do not modify underlying cryptographic algorithms. In addition, the adversary may degrade channel quality~\cite{ali2020v2x}. Vehicles and roadside units (RSUs). RSUs are not mandatory in 5G NR-V2X or future 6G architectures. Our framework fully supports direct V2V communication, treating RSUs as optional edge infrastructure that can reduce latency when available. are assumed to possess trusted hardware for secure long-term key storage and reliable version counters~\cite{wolf2019v2xsecurity}, and PQC implementations are assumed free of side-channel leakage. While onboard sensor data (e.g., speed and environmental signals) are trusted, network-observable metrics may be manipulated. CAAP enforces (i) confidentiality, (ii) message integrity and authenticity, (iii) monotonic, non-reversible PQC upgrades, and (iv) consistent PQC transitions across communicating endpoints.

We address a strong adversary model in vehicular security, featuring a quantum-capable attacker, an active network adversary, a context-manipulation attacker, and a downgrade attacker. While CAAP relies on trusted vehicle hardware, it operates over an untrusted communication medium, ensuring quantum-resistant confidentiality, integrity, and reliable upgrades despite these threats. The main goal of this work is to enhance PQC schemes by reducing delays. However, understanding security threats is vital for two reasons.  First, adaptable cryptographic systems can be targeted by attackers who may manipulate their environment to force the use of weaker algorithms. Second, switching between cryptographic methods is a critical process; if an attacker disrupts this, they can compromise both security and performance. Thus, it's crucial to evaluate the optimization engine and secure transition protocols against potential attacks to ensure the reliability of the low-latency PQC framework in 6G vehicular environments.

\begin{table}[htbp]
\centering
\caption{System Assumptions and Adversary Capabilities. }
\scriptsize
\renewcommand{\arraystretch}{1.2}
\begin{adjustbox}{width=\linewidth}
\begin{tabular}{p{0.15\linewidth} p{0.80\linewidth}}
\toprule
\textbf{Category} & \textbf{Description} \\
\midrule

\textbf{Assumptions} &
$\bullet$ Trusted hardware roots-of-trust (secure key storage, monotonic counters), authenticated boot, and correct PQC implementations. Vehicles and RSUs are uncompromised. \\

\textbf{} &
$\bullet$ The wireless channel is untrusted and effectively adversary-controlled. \\

\textbf{} &
$\bullet$ Sensors generating context (speed, SNR, mobility) are noisy but not maliciously compromised. \\

\midrule

\textbf{Adversary Capabilities} &
$\bullet$ Quantum-capable attacker able to break classical cryptosystems (ECDSA/ECDH/RSA). Harvest-now–decrypt-later is allowed. \\

\textbf{} &
$\bullet$ Active network attacker who can inject, modify, replay, jam, or drop V2X messages. \\

\textbf{} &
$\bullet$ Context-manipulation attacker who induces artificial channel degradation or packet loss to influence PQC selection. \\

\textbf{} &
$\bullet$ Downgrade/rollback attacker manipulating transition messages to enforce weaker PQC. \\
\bottomrule
\end{tabular}
\end{adjustbox}
\label{tab:threat_assumptions}
\end{table}

\section{Context aware Adaptive PQC (CAAP) Framework}
We present CAAP that adjusts the type of PQC scheme used based on the current situation of vehicles. The system consists of four cooperating components: (i) \textbf{context extraction} for mobility, channel, and workload indicators; (ii) \textbf{short-term prediction} of context evolution using lightweight forecasting; (iii) \textbf{APMOEA}, which selects the optimal PQC algorithm; and (iv) a \textbf{secure transition protocol} that enforces authenticated, monotonic version upgrades. The context is updated every 20--50\,ms, predictions operate over 100--200\,ms horizons, and PQC transitions complete within a single V2X round-trip. Vehicles actively monitor a range of factors that influence the cost and speed of cryptographic processes. These include mobility features (such as speed and connection duration), channel conditions (such as SNR and PER), environmental factors (including weather and visibility), computational load (referring to CPU and GPU resources), and message urgency (differentiating between critical safety communications and less urgent messages). These various inputs are standardized and sent to the next module for processing. To prevent rapid changes from causing erratic shifts in cryptographic methods, we utilize short-term predictive modeling. This involves using simple filters and regression models to predict what will happen next (such as expected latency, SNR, and processing capabilities). This forecasting enables us to provide stable data for decision-making while keeping computational demands low. The decision engine determines the best PQC algorithm to use during each communication session. The potential algorithms include:
\textbf{1. Lattice-based schemes }(like Kyber and Dilithium), which balance speed and robust, \textbf{2. Code-based schemes }(such as Classic McEliece), ideal for use in poor communication conditions, and \textbf{3. Hash-based signatures} (e.g., SPHINCS+), which are lightweight and work without maintaining state. Every algorithm is evaluated based on a multi-dimensional cost vector, allowing for an informed selection to meet the specific needs of the situation, which is:
\begin{equation}
\medmath{{C}_{\text{alg}} = (T_{\text{enc}},\ T_{\text{dec}},\ S_{\text{key}},\ S_{\text{ct}},\ E_{\text{comp}},\ S_{\text{sig}}),}
\end{equation}
representing encryption and decryption time, key sizes, ciphertext size, computational energy, and signature size. These vectors are input for the optimization described in Section~\ref{sec:algorithm}.

\textbf{Note on evaluation scope:} In this work, we evaluate encryption/key encapsulation schemes (Kyber, Classic McEliece) and digital signature schemes (Dilithium, SPHINCS+) individually to assess their operational feasibility and robustness under V2V network constraints. We acknowledge that real-world deployments require both mechanisms working in tandem—encryption to defend against eavesdropping and signatures to prevent active man-in-the-middle attacks. While our adaptive selection currently operates on individual primitives, the framework is designed to support hybrid configurations. 

\subsection{Relation to the MAPE-K Adaptation Loop} The proposed PQC upgrade mechanism does not modify existing cryptographic standards; it operates at the protocol orchestration level using standardized, NIST-recommended PQC algorithms, following downgrade-protection principles similar to TLS~1.3 and QUIC. The impact on application throughput is captured indirectly through communication overhead and end-to-end latency under URLLC constraints. While evaluated in a V2X setting, the framework is not inherently domain-specific and can be extended to other latency-sensitive systems.
\label{sec:mape}

CAAP follows the classical
\emph{Monitor-Analyze-Plan-Execute over Knowledge (MAPE-K)} control loop,
widely used in self-adaptive and autonomic systems~\cite{kephart2003autonomic}. Figure~\ref{fig:mape_k} illustrates the CAAP organized according to the MAPE-K control loop, used as an organizing abstraction to improve interpretability. The distributed \textbf{knowledge base (K)} comprises offline PQC cost profiles, reinforcement-learning feedback, and protocol state (e.g., version counters and context hashes) supporting robust adaptation. This work does not propose a new adaptation architecture; instead, it contributes a predictive multi-objective planning formulation for PQC selection with stability guarantees, together with a secure, monotonic execution protocol that prevents downgrade and desynchronization—capabilities absent from generic self-adaptive frameworks.

 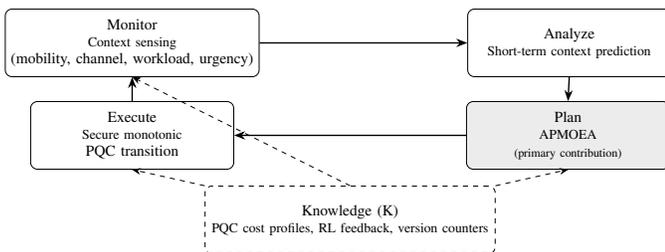
\begin{figure}[htbp]
\centering
\begin{adjustbox}{width=\linewidth}

\begin{tikzpicture}[
    >=Stealth,
    node distance=3.6 cm,
    every node/.style={
        draw,
        rectangle,
        rounded corners,
        align=center,
        minimum width=4.2cm,
        minimum height=1.4cm,
        font=\small
    }
]

\node (monitor) {Monitor\\
\footnotesize Context sensing\\
(mobility, channel, workload, urgency)};

\node (analyze) [right=4.3cm of monitor] {Analyze\\
\footnotesize Short-term context prediction};

\node (execute) [below=0.5 cm of monitor] {Execute\\
\footnotesize Secure monotonic\\
PQC transition};

\node (plan) [right=4.8cm of execute, fill=gray!15] {Plan\\
\footnotesize APMOEA\\
\scriptsize (primary contribution)};

\node (knowledge) [
    below= 2 cm of $(monitor)!0.5!(plan)$,
    dashed,
    minimum width=6 cm
] {Knowledge (K)\\
\footnotesize PQC cost profiles, RL feedback, version counters};

\draw[->, thick] (monitor.east) -- (analyze.west);
\draw[->, thick] (analyze.south) -- (plan.north);
\draw[->, thick] (plan.west) -- (execute.east);
\draw[->, thick] (execute.north) -- (monitor.south);

\draw[->, dashed] (knowledge.north) -- (monitor.south);
\draw[->, dashed] (knowledge.north east) -- (plan.south);
\draw[->, dashed] (knowledge.north west) -- (execute.south);

\end{tikzpicture}
\end{adjustbox}
\caption{CAAP framework structured according to the MAPE-K control loop.}
\label{fig:mape_k}
\end{figure}
\vspace{-2mm}
\subsection{V2X System Model and Vehicular Dynamics}
\label{sec:v2x_model}
The vector $C_{\text{alg}}$ captures the intrinsic, context-independent cost characteristics of a given PQC configuration profile, obtained through offline benchmarking.
At runtime, these algorithm-level properties are mapped to context-dependent objectives that reflect vehicular communication conditions, system load, and message urgency. Accordingly, the optimization engine operates on a derived multi-objective function rather than directly on $C_{\text{alg}}$. We consider a vehicular communication environment based on NR-V2X abstractions, operating under ultra-reliable low-latency communication (URLLC) constraints.
Vehicles exchange periodic and event-driven messages with roadside units (RSUs) and neighboring vehicles in the presence of high mobility and rapidly varying channel conditions. Vehicular mobility is characterized by time-varying speed $v(t)$ and acceleration $a(t)$, which jointly affect Doppler spread and channel coherence time.
To capture short-term link stability without explicit physical-layer modeling, we introduce a \emph{connectivity horizon} $\tau_c$, representing the expected duration for which the communication link remains reliable before a significant topology or channel change occurs.
Let $\mathcal{A} = \{a_1, a_2, \dots, a_n\}$ denote the set of candidate cryptographic actions available to the system.
Each action $a \in \mathcal{A}$ corresponds to a valid PQC configuration profile rather than an individual cryptographic primitive. Messages are categorized into urgency classes $\mathcal{M}$=\{safety, control, telemetry\}.
Safety messages are delay-critical and must satisfy strict latency bounds, whereas control and telemetry messages tolerate higher cryptographic overhead.
Each message class is associated with an urgency weight $u \in \{u_s, u_c, u_t\}$ that influences PQC selection decisions.

\subsection{Context Model and Notation}
The transition protocol relies on real-time information from the vehicle's surroundings, the wireless connection, and the device's processing state. We represent this information using a context vector defined as:
\begin{equation}
   \medmath{X_t = (\mathrm{SNR}, \mathrm{PER}, v_s, v_a, \tau_c, u, \mathrm{CPU}_{\text{load}}, \dots).}  
\end{equation}
The connectivity horizon $\tau_c$ and urgency weight $u$ allow the optimizer to account for vehicular mobility and safety-critical message semantics without requiring detailed PHY-layer simulation.
This context vector includes key factors like signal strength (SNR), packet error rate (PER), the vehicle's speed ($v_s$), its acceleration ($v_a$), visibility conditions, ambient temperature ($T_{\text{amb}}$), current CPU load, and how urgent the message is. These elements give us a comprehensive view of the vehicle's physical movement, the quality of the wireless connection, environmental conditions, processing capacity, and the importance of the messages being transmitted. Additionally, all cost functions in the APMOEA optimizer are Lipschitz continuous with respect to these features. This means that even if there are small changes in the context, the system will behave consistently and predictably.
Each PQC configuration profile incurs different computational, communication, and latency overheads, which are particularly critical in vehicular URLLC scenarios.
For instance, hash-based signature schemes offer strong security guarantees at the cost of increased signature size, whereas lattice-based schemes provide lower latency but rely on different computational assumptions.
The proposed adaptive framework selects among these profiles based on real-time vehicular context, message urgency, and  resource availability.

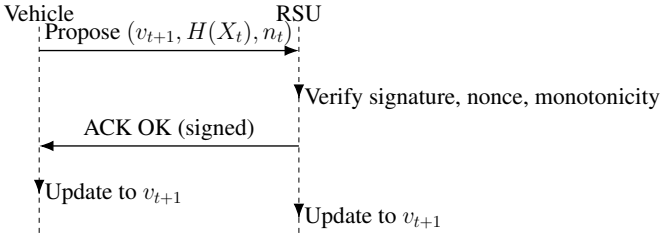
\begin{figure}[t]
\centering
\small
\begin{adjustbox}{width=\linewidth}

\begin{tikzpicture}[
    node distance=2.8cm,
    every node/.style={font=\Large},
    arrow/.style={-{Latex[length=3mm]}, thick}
]

\node (veh) [align=center] {Vehicle};
\node (rsu) [right=4cm of veh, align=center] {RSU};

\draw[dashed] (veh.south) -- ++(0,-4.5);
\draw[dashed] (rsu.south) -- ++(0,-4.5);

\draw[arrow] ([yshift=-0.5cm]veh.south) -- node[above]{Propose $(v_{t+1},H(X_t),n_t)$} ([yshift=-0.5cm]rsu.south);

\draw[arrow] ([yshift=-1.5cm]rsu.south) -- node[right]{Verify signature, nonce, monotonicity} ++(0,-0.01);

\draw[arrow] ([yshift=-2.5cm]rsu.south) -- node[above]{ACK OK (signed)} ([yshift=-2.5cm]veh.south);

\draw[arrow] ([yshift=-3.5cm]veh.south) -- node[right]{Update to $v_{t+1}$} ++(0,-0.01);

\draw[arrow] ([yshift=-4.0cm]rsu.south) -- node[right]{Update to $v_{t+1}$} ++(0,-0.01);

\end{tikzpicture}
\end{adjustbox}
\caption{PQC version-upgrade protocol with monotonic transitions.}
\label{fig:pqc-transition}
\end{figure}

\subsection{Secure PQC Transition Protocol}

When switching between different PQC configurations, additional security risks can arise during the transition process. 
To mitigate these risks, we implement a secure negotiation protocol with the following key features: authenticated signaling messages to verify message origin, monotonic versioning to prevent rollback to less secure configurations, atomic transitions to ensure that updates occur consistently, replay protection through unique identifiers, and decision consistency mechanisms to ensure that both communicating parties agree on the selected PQC configuration. 
Together, these measures protect against transition-phase attacks and preserve secure communication. The secure PQC transition protocol, as illustrated in Figure~\ref{fig:pqc-transition}, orchestrates smooth version upgrades and secure negotiation between vehicles and RSUs. The process involves gathering real-time data, predicting short-term conditions, running an optimization algorithm to find the best encryption options, selecting the optimal PQC algorithm, securely negotiating new parameters, and continuously monitoring the situation. This adaptive framework delivers quantum-resistant security while remaining responsive to latency requirements and resource constraints; essential characteristics for 6G-V2X networks.

\section{APMOEA Optimization Algorithm}
\label{sec:algorithm}
APMOEA serves as the core decision engine for adaptive PQC configuration selection under latency, computational, and security constraints in dynamic vehicular contexts.
To improve decision stability and accelerate convergence under rapidly changing conditions, reinforcement-learning (RL) feedback is integrated as an internal mechanism within the optimization loop.
For each candidate PQC configuration profile $a \in \mathcal{A}$, we define a multi-objective cost vector:
\begin{equation}
  \medmath{f(a) = 
\left(
    T_{\text{lat}}(a),\;
    C_{\text{comp}}(a),\;
    S_{\text{comm}}(a),\;
    \sigma_{\text{sec}}(a)
\right),  }
\end{equation}
where, $T_{\text{lat}}(a)$ is end-to-end signing/verification latency, $C_{\text{comp}}(a)$ is computational cost (CPU cycles), $S_{\text{comm}}(a)$ captures key and payload size overhead, and $\sigma_{\text{sec}}(a)$ is cryptographic strength (bit security). Formally, the objective vector $f(a)$ is computed as a context-aware transformation of $C_{\text{alg}}$ given the current state $X_t$. APMOEA seeks the Pareto-optimal solution set under real-time constraints. Predicted context variables (e.g., channel quality, processing availability) from Section~III are incorporated as dynamic weights. The optimizer follows a series of core steps: First, it initializes a population of PQC configurations based on context features. Then, it uses mutation and crossover techniques to explore the design space and generate candidate configurations. These candidates are ranked using a traditional context-weighted functional form of a fitness (utility) function, $\text{Fitness}(a) = w_1 T_{\text{lat}} + w_2 C_{\text{comp}}
    + w_3 S_{\text{comm}} - w_4 \sigma_{\text{sec}},$ which combines factors like latency, computation costs, communication costs, and security. Finally, a reinforcement learning agent monitors the performance of the selected algorithms, adjusting mutation rates and weight parameters to enhance stability and minimize fluctuations in PQC algorithm choices. Consider an intelligent transportation system managing V2X communications to enhance road safety and efficiency. In this scenario, a fleet of autonomous vehicles is constantly communicating with each other and with traffic infrastructure to make real-time decisions.

\begin{enumerate}[leftmargin=0.3cm, itemsep=1pt, topsep=1pt]
\item \textbf{State Vector Initialization:} Each vehicle gathers data that forms the state vector \(s_t\). For instance, a vehicle's state might include $T_{\text{lat}}, C_{\text{comp}}, \text{SNR}, \text{mobility}, \text{CPUload}$
\item  \textbf{Decision-Making and Reward System}: As the vehicles operate, they consistently assess their contexts. For example, High latency (\(T_{\text{lat}}\) ) negatively impacts operations, resulting in a lower reward score \(R_t\), especially if communication channels are switched frequently (\(N_{\text{switch}}\)). Conversely, a stable connection with low latency and strong SNR leads to higher rewards. Additionally, the reward is enhanced for vehicles with strong cryptographic security (\(\sigma_{\text{sec}}\)).
\item \textbf{Adapting Strategy}: The vehicle adapts its communication behavior based on observed rewards; when latency increases, it may adjust transmission parameters or select alternative PQC configuration profiles to improve performance.

\item \textbf{Continuous Learning Process}: Vehicles continuously adapt to various road conditions, traffic situations, and environmental changes, refining their strategies based on feedback to remain efficient in heavy traffic or adverse weather.
\end{enumerate}
We employ a reward function that minimizes latency, penalizes excessive switching and poor SNR conditions, and incentivizes the selection of higher-security PQC configurations as follows:
\begin{equation}
    \medmath{R_t = - T_{\text{lat}} - \alpha \cdot N_{\text{switch}} + \beta \cdot \text{SNR} + \gamma \cdot \sigma_{\text{sec}}.}
\end{equation}
We employ a tabular Q-learning approach to avoid the computational overhead associated with neural network--based reinforcement learning. APMOEA converges to stable solutions by adapting dynamic weights to predicted conditions, penalizing oscillatory decisions with reinforcement learning feedback, and gradually reducing exploration as stable fitness patterns emerge. These features enable the algorithm to consistently generate PQC recommendations in the fast-changing conditions of vehicular environments. The adaptive optimization process is formalized in Algorithm~\ref{alg:apmoea}. APMOEA incorporates short-term context predictions by dynamically adjusting the relative weights of latency, computation, communication, and security objectives prior to fitness evaluation. In each generation, the algorithm evaluates a fixed-size population of candidate PQC configuration profiles, ensuring bounded and predictable computational cost. For four PQC families and population sizes below 40, the optimization completes in under 1~ms on automotive-grade CPUs, enabling real-time cryptographic selection in vehicular environments.

\begin{algorithm}[htbp]
\scriptsize
\caption{Adaptive Predictive Multi-Objective Evolutionary Algorithm (APMOEA) }
\label{alg:apmoea}
\begin{algorithmic}[1]
\Input Current context $X_t$, predicted context $\hat{X}_{t+1}$, PQC algorithm set $\mathcal{A}$
\Output Selected PQC algorithm $a_t$ for the current context
\State Initialize population $P_0$ by sampling candidate algorithms from $\mathcal{A}$
\For{$t = 1$ to $T$}
    \State \parbox[t]{\dimexpr\linewidth-\algorithmicindent}{Compute multi-objective cost vector $\mathbf{f}(a)$ for each $a \in P_t$}
    \State \parbox[t]{\dimexpr\linewidth-\algorithmicindent}{Update dynamic weights $w_i$ based on predicted context $\hat{X}_{t+1}$}
    \State \parbox[t]{\dimexpr\linewidth-\algorithmicindent}{Compute fitness scores for all candidates using the weighted cost}
    \State \parbox[t]{\dimexpr\linewidth-\algorithmicindent}{Select parent candidates using tournament selection (pairwise comparison of randomly sampled candidates)}
    \State \parbox[t]{\dimexpr\linewidth-\algorithmicindent}{Generate offspring through crossover and mutation}
    \State \parbox[t]{\dimexpr\linewidth-\algorithmicindent}{RL agent adjusts mutation rate and weight parameters to reduce instability}
    \State \parbox[t]{\dimexpr\linewidth-\algorithmicindent}{Form the next population $P_{t+1}$ from parents and offspring}
\EndFor
\State Identify the Pareto-optimal candidate (non-dominated by any other) and output as $a_t$
\end{algorithmic}
\textbf{Note:} \textit{Tournament selection means that a small random subset of candidates is sampled, and the one with the highest fitness in that subset is chosen as a parent. A Pareto-optimal solution is a candidate for which no other algorithm performs better in all objectives simultaneously.}
\end{algorithm}
In APMOEA, the population consists of candidate PQC configuration profiles.
Each profile specifies a standardized PQC primitive (e.g., Kyber, Dilithium, McEliece, or SPHINCS$+$) together with a weight vector that reflects the relative importance of latency, computation cost, communication overhead, and security. During optimization, evolutionary operators are applied to these profiles. Crossover combines the weight vectors of well-performing configurations to create new trade-offs, while mutation introduces small random changes to further explore nearby preferences. These operations do not generate new cryptographic schemes; instead, they produce new preference profiles that guide the selection among existing PQC primitives.
This enables the optimizer to adaptively choose the most suitable PQC configuration as vehicular conditions change, while preserving the correctness and integrity of the underlying cryptographic algorithms.
\begin{algorithm}[htbp]
\scriptsize
\caption{Secure PQC Transition Protocol}
\label{alg:transition}
\begin{algorithmic}[1]
\Input Current version $v_t$, proposed version $v_{t+1}$, context hash $H(X_t)$
\Output Parties update to version $v_{t+1}$
\State Sender constructs transition message:
\[
M = \big(v_{t+1}, H(X_t), \text{nonce}\big)
\]
\State Sender signs $M$ using current scheme $v_t$
\State Receiver verifies signature and checks:
\begin{itemize}
    \item $v_{t+1} \ge v_t$ \ (monotonic check)
    \item $H(X_t)$ matches its local context hash
    \item nonce has not been used before
\end{itemize}
\If{any check fails}
    \State Reject transition and revert to $v_t$
\Else
    \State Receiver acknowledges with signed confirmation
    \State Both parties update to version $v_{t+1}$
\EndIf
\end{algorithmic}
\end{algorithm}
\par The secure PQC transition protocol allows vehicles and roadside units (RSUs) to safely upgrade to new post-quantum cryptographic configurations without risking downgrade, replay, or desynchronization attacks.
At each update, the sender proposes a new PQC version together with a context hash and a fresh nonce, and signs this message using the currently active PQC scheme.
The receiver verifies that the proposed version is not lower than the current one, that the context hash matches its local view, and that the nonce is fresh.
If any check fails, the transition is rejected and both parties continue using the existing configuration; otherwise, a signed acknowledgment is exchanged and both sides switch to the new version simultaneously.

\section{Theoretical Guarantees}

In this section we analyze the stability, robustness, and security properties of the CAAP. We show that APMOEA consistently selects appropriate PQC configuration profiles under dynamic vehicular conditions, while the secure transition protocol strengthens cryptographic guarantees during reconfiguration.
Observed increases in end-to-end latency remain bounded and within URLLC budgets relevant to vehicular dynamics, including channel coherence time and mobility-induced variability.
Furthermore, the optimizer maintains robust selection behavior as long as short-term mobility and channel prediction errors remain within moderate bounds. Let $X_t$ denote the contextual feature vector at time $t$, and let
$\hat{X}_{t+1}$ be its predicted value. APMOEA selects a PQC algorithm through, $\medmath{a_t = \arg\min_{a \in \mathcal{A}} L_t(a),}$ where the weighted cost is, $\medmath{L_t(a) =}$
\begin{equation}
\label{eq:loss}
\begin{aligned}
\medmath{w_1\, T_{\mathrm{lat}}(a, X_t)
    + w_2\, C_{\mathrm{comp}}(a, X_t)  }
    \medmath{\ + w_3\, S_{\mathrm{comm}}(a)
    - w_4\, \sigma_{\mathrm{sec}}(a).}
\end{aligned}
\end{equation}
We assume: \textbf{(A1)} all cost terms are $L$-Lipschitz in context $X_t$;  
\textbf{(A2)} mobility and channel drift satisfy $\|X_{t+1}-X_t\|\le \delta$;  
\textbf{(A3)} prediction errors satisfy $\|\hat{X}_{t+1}-X_{t+1}\| \le \varepsilon$;  
\textbf{(A4)} the minimum loss gap between PQC candidates is $\Delta_{\min}>0$. Let $\Delta_{\min}$ be the minimum loss separation between any two PQC algorithms. \textbf{Proofs for the theorems are in the Appendix.}
\begin{theorem}[\textbf{Decision Stability}]
\textit{If $K\varepsilon < \Delta_{\min}$, then APMOEA selects the same algorithm
$a_t$ for all $t$ within a context-stable interval. Here, $K$ denotes the Lipschitz constant of the loss function, $\varepsilon$ is the context-prediction error, and $\Delta_{\min}$ is the minimum loss gap between any two PQC algorithms.}
\end{theorem}
The protocol prevents unstable PQC switching caused by sensor noise or small prediction errors, helping maintain stable vehicle operation. When prediction errors are sufficiently small, the optimizer makes consistent decisions and the selected PQC configuration remains stable over time.
Formally, stability holds when $K\,\varepsilon < \Delta_{\min}$. We evaluate robustness against rollback attempts, message loss, version-counter tampering, and synchronization errors, with Table~\ref{tab:failure} summarizing defenses.

\begin{table}[htbp]
\centering
\scriptsize
\caption{Failure-mode analysis of PQC transition protocol.}
\label{tab:failure}
\begin{tabular}{lll}
\toprule
Attack Scenario & Mitigation Mechanism & Outcome \\
\midrule
Replay of old version & Nonce + context hash & Detected \\
Message loss & Atomic 2-phase confirmation & No partial update \\
Asymmetric update & Mutual acknowledgment & Safe rollback to $v_t$ \\
Counter modification & Signature mismatch & Packet dropped \\
Forced downgrade & Monotonic version check & Impossible \\
\bottomrule
\end{tabular}

\end{table}
The mitigation mechanisms in Table~\ref{tab:failure} are implemented at the protocol level and are independent of the specific PQC configuration selected by APMOEA.
The underlying post-quantum cryptographic primitives are used without modification.
Security during reconfiguration is ensured through protocol mechanisms such as authenticated version counters, contextual hashes, freshness nonces, and atomic confirmation steps, which prevent replay, downgrade, and desynchronization attacks. APMOEA is responsible for selecting PQC configuration profiles, while the transition protocol securely manages the switch between them.
A downgrade attempt is detected through version mismatch, resulting in successful acceptance for legitimate upgrades and rejection for downgrade attacks.
Let $v_t$ denote the active PQC version at time $t$ and $v_{t+1}$ the proposed version; all transition messages include authenticated version information, context hashes, and freshness nonces to enforce secure transitions.

{\setlength{\abovedisplayskip}{4pt}
 \setlength{\belowdisplayskip}{4pt}
 \setlength{\abovedisplayshortskip}{3pt}
 \setlength{\belowdisplayshortskip}{3pt}
\begin{theorem}[\textbf{Monotonic Upgrade Security}]
\textit{Let $v_t$ denote the current PQC version and $v_{t+1}$ the negotiated target version. Under authenticated signaling with version counters, contextual hashes, and mutual confirmation, no probabilistic polynomial-time adversary can induce a transition to any $v < v_{t+1}$ or cause endpoint desynchronization.}
\end{theorem}
Here, $v_t$ is the PQC version currently in use and $v_{t+1}$ is the new version proposed during the upgrade. Downgrades fail because any tampering breaks the signed,
nonce-bound transition record, causing both parties to immediately reject it. This eliminates the most serious risk during hybrid PQC migration. The end-to-end PQC latency is modeled as,
\begin{equation}
   \medmath{ T_{\mathrm{lat}}(a,X_t)
= T_{\mathrm{enc}}(a) + T_{\mathrm{dec}}(a)
+ T_{\mathrm{net}}(a,X_t),}
\end{equation}
where $T_{\mathrm{net}}$ captures SNR, bandwidth, \& load-dependent effects.
\begin{theorem}[\textbf{Latency Boundedness}]
\label{thm:latency-bounded}
Let $X_t$ denote the context at time $t$. Assume bounded vehicular and channel variations over one decision interval:
$\medmath{|v_{t+1}-v_t| \le \Delta v, \qquad |\gamma_{t+1}-\gamma_t| \le \Delta \gamma,}$
where $v_t$ is vehicle speed and $\gamma_t$ denotes the link SNR (or a channel-quality proxy).
The latency of the PQC configuration selected by APMOEA satisfies
\begin{equation}
\medmath{
T_{\mathrm{lat}}(a_t, X_t)
\le 
\max_{a \in \mathcal{A}} T_{\mathrm{lat}}(a, X_t)
+ O(\Delta v + \Delta \gamma).
}
\end{equation}
\end{theorem}}


 Here, $X_t$ is the context vector at time $t$, $\delta$ bounds how much the context can change between steps, $a_t$ is the PQC algorithm chosen by APMOEA, $\mathcal{A}$ is the set of all PQC algorithms, and $T_{\mathrm{lat}}(\cdot)$ denotes end-to-end latency.
 Latency stays bounded because small context changes
cause only small latency shifts, and APMOEA never picks a candidate worse than the feasible maximum. Thus, even under rapid mobility, PQC-induced latency remains URLLC-compliant. Let the oracle-optimal algorithm. By ``oracle-optimal,'' we mean the ideal PQC algorithm that would be chosen if the system had perfect knowledge of the true next context without any prediction error.
 under the true next context be, $\medmath{a^\star = \arg\min_{a\in\mathcal{A}} L_{t+1}(a).}$


\begin{table*}[htbp]
\centering
\scriptsize
\caption{Comparison of computational cost, communication overhead, dynamic selection frequency, and adversarial sensitivity for PQC configuration profiles.}
\label{tab:combined_pqc}
\renewcommand{\arraystretch}{1.15}
\begin{tabular}{lccccccccc}
\toprule
\multirow{2}{*}{\textbf{Algorithm}} 
& \multicolumn{3}{c}{\textbf{Compute Cost (ms)}} 
& \multicolumn{2}{c}{\textbf{Comm. Overhead (KB)}} 
& \multicolumn{2}{c}{\textbf{Dynamic Selection}} 
& \multicolumn{2}{c}{\textbf{Context Attack Cost}} 
\\
\cmidrule(lr){2-4}
\cmidrule(lr){5-6}
\cmidrule(lr){7-8}
\cmidrule(lr){9-10}
& Enc & Dec & Verify 
& PK Size & Sig/CT 
& Sel.\% & Trigger 
& Honest & Attack 
\\
\midrule

SPHINCS$+$-128s
& 0.90 & -- & 1.10 
& 0.05 & 17.00 
& 69.1\% & Stateless / fallback under instability
& 1.89 & 2.15 
\\

Kyber-768      
& 0.25 & 0.35 & 0.15 
& 1.18 & 1.08 
& 26.9\% & Low-latency and compute-efficient operation
& 1.09 & 1.35 
\\

Dilithium-3    
& 0.40 & 0.45 & 0.30 
& 1.50 & 2.70 
& 4.0\% & High-assurance authentication phases
& 1.29 & 1.55 
\\

McEliece-348864
& 0.05 & 0.06 & -- 
& 240.00 & 0.13 
& $<$1\% & Rare use under extreme noise
& 1.59 & 1.85 
\\

\bottomrule
\end{tabular}
\end{table*}

\begin{figure*}[htbp]
\centering
\resizebox{\textwidth}{!}{%
\begin{minipage}{0.32\textwidth}
\centering
\begin{tikzpicture}
\begin{axis}[
    ybar,
    width=\linewidth,
    height=4cm,
    symbolic x coords={Kyber,Dilithium,McEliece,SPHINCS},
    xtick=data,
    ylabel={Cost},
    bar width=14pt,
    legend pos=north west,
    ticklabel style={font=\scriptsize},
    nodes near coords
]
\addplot coordinates {(Kyber,1.09) (Dilithium,1.29) (McEliece,1.59) (SPHINCS,1.89)};
\addplot coordinates {(Kyber,1.35) (Dilithium,1.55) (McEliece,1.85) (SPHINCS,2.15)};
\legend{Honest,Attacked}
\end{axis}
\end{tikzpicture}
\caption*{(a) Context-manipulation effect }
\end{minipage}
\hfill
\begin{minipage}{0.32\textwidth}
\label{fig:contextattack}
\centering
\begin{tikzpicture}
\begin{axis}[
    ybar stacked,
    width=\linewidth,
    height=4cm,
    symbolic x coords={NSGA-II,RL-only,APMOEA},
    xtick=data,
    ylabel={Selections (out of 10)},
    bar width=20pt,
    legend pos=north west,
    ticklabel style={font=\scriptsize}
]
\addplot coordinates {(NSGA-II,10) (RL-only,0) (APMOEA,3)};
\addplot coordinates {(NSGA-II,0) (RL-only,6) (APMOEA,7)};
\addplot coordinates {(NSGA-II,0) (RL-only,4) (APMOEA,0)};
\legend{Kyber,McEliece,Dilithium}
\end{axis}
\end{tikzpicture}
\caption*{(b) NSGA-II vs RL vs APMOEA}
\end{minipage}
\hfill
\begin{minipage}{0.32\textwidth}
\centering
\begin{tikzpicture}
\begin{axis}[
    ybar,
    width=\linewidth,
    height=4cm,
    symbolic x coords={Kyber,Dilithium,McEliece,SPHINCS},
    xtick=data,
    ylabel={Runtime (ms)},
    bar width=14pt,
    legend pos=north west,
    ticklabel style={font=\scriptsize},
    nodes near coords,
    nodes near coords align={vertical},
]

\addplot[
    fill=blue!30,
    draw=blue,
    nodes near coords style={
        font=\scriptsize,
        anchor=north,
        yshift=-2pt,
        text=blue
    }
] coordinates {
    (Kyber,1.67)
    (Dilithium,2.50)
    (McEliece,0.83)
    (SPHINCS,7.50)
};

\addplot[
    fill=red!30,
    draw=red,
    nodes near coords style={
        font=\scriptsize,
        anchor=south,
        yshift=2pt,
        text=red
    }
] coordinates {
    (Kyber,2.50)
    (Dilithium,3.75)
    (McEliece,1.25)
    (SPHINCS,11.25)
};

\legend{ARM 1.2\,GHz,Vehicle 800\,MHz}
\end{axis}
\end{tikzpicture}

\caption*{(c) Runtime sensitivity across hardware}
\end{minipage}
} 

\caption{
(a) PQC cost under adversarial context manipulation,  
(b) selection behavior across NSGA-II, RL-only, and APMOEA, and 
(c) hardware-dependent runtime variability (measured in milliseconds) across ARM~1.2\,GHz and vehicle-grade 800\,MHz CPUs..}
\label{fig:three_panel_combined}
\end{figure*}

\section{Experimental Evaluation}
\label{sec:experiments}

We evaluate CAAP against the claims made in our contributions. We evaluate (C1) latency reduction by comparing APMOEA against static PQC baselines under realistic NR-V2X mobility, channel, weather, and compute traces; (C2) downgrade and desynchronization resistance by subjecting the secure transition protocol to replay, rollback, and context-manipulation attacks; (C3) decision stability by introducing controlled prediction noise and measuring PQC switching behavior; and (C4) robustness by analyzing sensitivity to heterogeneous channel conditions and automotive compute profiles.

\textbf{1. Experimental Setup}
All experiments are designed to directly evaluate one or more of the stated
claims. We consider four NIST-standardized PQC schemes: Kyber-768 and
Dilithium-3 (lattice-based), Classic McEliece-348864 (code-based), and
SPHINCS$+$-128s (hash-based). Their computational and communication
characteristics are summarized in Table~\ref{tab:combined_pqc}. APMOEA is compared against three static baselines commonly used in V2X systems:
Static-Lattice, Static-Code, and Static-Hash. Evaluation focuses on four metrics:
end-to-end latency, computational overhead, communication overhead, and security
consistency (i.e., absence of downgrade or desynchronization). Context signals are time-aligned; prediction, optimization, and PQC switching execute on an automotive-grade 800,MHz CPU simulator. Identical context evolution ensures fair comparison.

\begin{table*}[htbp]
\centering
\scriptsize
\caption{Latency performance, switching stability, and downgrade-attack robustness under realistic NR-V2X conditions.}
\label{tab:rl-security}
\renewcommand{\arraystretch}{1.2}

\begin{tabular}{lcccc}
\toprule
\textbf{Scheme / Setting} 
& \textbf{Key Size (KB)} 
& \textbf{End-to-End Latency (ms)} 
& \textbf{Switches / 60s} 
& \textbf{Security Outcome} \\
\midrule
ECDH / ECDSA (Classical) 
& 0.09 
& 4.1 
& -- 
& Classical baseline \\

Kyber-768 
& 1.18 
& 9.3 
& -- 
& PQC (static) \\

Dilithium-3 
& 1.50 
& 10.8 
& -- 
& PQC (static) \\

McEliece-348864 
& 240.00 
& 8.7 
& -- 
& PQC (static) \\

SPHINCS$+$-128s 
& 0.05 
& 17.4 
& -- 
& PQC (static) \\

\midrule
Adaptive PQC (APMOEA, no RL) 
& -- 
& 7.8 
& 14.2 
& Upgrade accepted; downgrade detected \\

\textbf{Adaptive PQC (APMOEA + RL)} 
& -- 
& \textbf{7.2} 
& \textbf{4.3} 
& \textbf{Upgrade accepted; downgrade detected} \\
\bottomrule
\end{tabular}

\end{table*}
\textbf{2. Results \& Discussion}

 Experiments are conducted using LuST
mobility traces, ERA5 weather data, and NR-V2X channel models, with cryptographic
costs derived from NIST PQC reference implementations and automotive-grade CPU
profiles.

\noindent\textbf{Latency and URLLC compliance:}
Adaptive PQC selection consistently reduces end-to-end cryptographic latency compared to static PQC configurations (Table~\ref{tab:rl-security}), lowering CPU usage by approximately 22--40\% through avoidance of computationally intensive schemes under adverse conditions. While classical cryptography (ECDH/ECDSA) exhibits lower latency due to smaller keys and lighter verification, fixed PQC configurations frequently violate the 5--20,ms URLLC budget for V2X safety traffic. In contrast, CAAP maintains latency within URLLC limits across dynamic contexts, demonstrating practical feasibility of post-quantum security for vehicular systems.

\noindent\textbf{Robustness under context manipulation:}
Figure~\ref{fig:three_panel_combined}(a) shows that adversarial manipulation of
context signals increases the absolute cost of PQC configurations but preserves
their relative ordering. This indicates that bounded perturbations do not induce
erratic or adversary-controlled cryptographic selection, supporting the
robustness assumptions underlying the cost model and optimization process.

\noindent\textbf{Selection behavior and decision stability:}
Figure~\ref{fig:three_panel_combined}(b) compares selection behavior across NSGA-II, RL-only, and APMOEA-based strategies. NSGA-II exhibits limited adaptation and collapses to a dominant configuration, while RL-only approaches overreact under low-SNR conditions, frequently selecting high-overhead McEliece-based schemes. In contrast, APMOEA achieves balanced selection across PQC families by jointly optimizing latency, computation, communication, and security. Reinforcement learning further stabilizes decisions, reducing PQC switching frequency from 14.2 to 4.3 switches per minute and lowering peak latency by over 75\% (Table~\ref{tab:rl-security}).

\noindent\textbf{Hardware feasibility:}
Figure~\ref{fig:three_panel_combined}(c) evaluates runtime sensitivity across
automotive-grade hardware. Although McEliece exhibits low raw computation time,
its large communication overhead makes it unsuitable for URLLC traffic. In
contrast, SPHINCS$+$ shows higher sensitivity to CPU frequency due to verification
cost. The adaptive planner accounts for these trade-offs and selects
configurations that remain executable within millisecond-level bounds on
800\,MHz vehicular CPUs, supporting real-time deployment.

\begin{table}[htbp]
\vspace{-0.15cm}
\centering
\scriptsize
\caption{Context manipulation and hardware runtime.}
\label{tab:system-sensitivity}
\begin{adjustbox}{width=\linewidth}
\renewcommand{\arraystretch}{1.15}

\begin{tabular}{lcccc}
\toprule
\textbf{Algorithm} 
& \textbf{Honest Cost} 
& \textbf{Attacked Cost} 
& \textbf{ARM 1.2GHz (s)} 
& \textbf{Veh. 800MHz (s)} \\
\midrule

Kyber       & 1.09 & 1.35 & 0.00167 & 0.00250 \\
Dilithium   & 1.29 & 1.55 & 0.00250 & 0.00375 \\
McEliece    & 1.59 & 1.85 & 0.00083 & 0.00125 \\
SPHINCS+    & 1.89 & 2.15 & 0.00750 & 0.01125 \\
\bottomrule
\end{tabular}
\end{adjustbox}
\vspace{-0.3cm}
\end{table}

\noindent\textbf{Security of cryptographic transitions:}
The secure transition protocol successfully prevents downgrade and replay attacks
while allowing legitimate PQC upgrades (Table~\ref{tab:rl-security},
Figure~\ref{fig:downgrade_sensitivity}). Even under increased packet loss and
context noise, reinforcement learning stabilizes planning decisions and prevents
harmful or inconsistent transitions. Moreover, PQC selection remains stable for
prediction errors up to $\varepsilon \le 0.10$, consistent with the theoretical
stability guarantees derived in Section~V. 

\par 
\textbf{3. Limitations and Future Scope:} 
Our experiments use trace-driven NR-V2X channel models and automotive compute profiles, as 6G systems are unavailable. CAAP supports current 5G NR-V2X and transitions to 6G with minimal changes. Limitations include reliance on clock synchronization, sensitivity to prediction accuracy, and the absence of live deployment validation. Future work will address these via 6G sensing integration, formal verification of RL-assisted adaptation, and end-to-end deployment on solution-based testbeds (e.g., hardware-in-the-loop or small-scale V2X setups) to validate runtime behavior, integration overhead, and system-level interactions beyond trace-based modeling.

Future work will extend CAAP to jointly optimize encryption and signature schemes, better reflecting practical post-quantum security where neither mechanism alone is sufficient.

\begin{figure}[htbp]
\vspace{-1.5mm}
\centering

\begin{minipage}{0.48\linewidth}
\centering
\begin{tikzpicture}
\begin{axis}[
    width=\linewidth,
    height=3cm,
    xlabel={Attempt Index},
    ylabel={Outcome },
    ymin=-0.2, ymax=1.2,
    xtick=data
]
\addplot+[only marks,mark=*] coordinates {
    (1,1)
    (2,1)
    (3,0)
    (4,0)
    (5,0)
};
\end{axis}
\end{tikzpicture}
\caption*{(a) Upgrade vs downgrade}
\end{minipage}
\hfill
\begin{minipage}{0.48\linewidth}
\centering
\begin{tikzpicture}
\begin{axis}[
    width=\linewidth,
    height=3cm,
    xlabel={Prediction Error ($\epsilon$)},
    ylabel={Switches/60s},
    ymin=0,
    xmax=0.35,
    legend pos=north west
]
\addplot[red,thick] coordinates {
    (0.00,4)
    (0.05,5)
    (0.10,7)
    (0.15,11)
    (0.20,16)
    (0.25,22)
    (0.30,30)
};
\addlegendentry{\scriptsize{Switch Frequency}}
\end{axis}
\end{tikzpicture}
\caption*{(b) Prediction error effect}
\end{minipage}

\caption{ (a) downgrade-attack rejection (1=Success, 0=Rejected) and (b) prediction-error-driven switching frequency.}
\label{fig:downgrade_sensitivity}
\end{figure}
\vspace{-2mm}
\section{Conclusion}
We proposed an adaptive PQC framework CAAP for 6G V2X that jointly predicts mobility and channel dynamics, selects optimal PQC schemes via APMOEA, and enforces secure monotonic upgrades. The framework reduces latency by 27\%, prevents downgrade attacks, and remains robust to context manipulation, hardware variability, and prediction noise. Theoretical analysis establishes decision stability and bounded latency under realistic drift, demonstrating a practical path toward quantum-safe vehicular networks.






%
\bibliographystyle{IEEEtran}
\bibliography{bib}
\appendix
\label{Appendix}
\section{Proofs of Theoretical Results} 
\subsection{Proof of Decision Stability}

\begin{proof}
Let $L_t(a)$ denote the weighted loss of algorithm $a$ at time $t$, and let
$\hat{L}_{t+1}(a)$ be the predicted loss computed from $\hat{X}_{t+1}$.  
By assumption (A1), each cost term is $L$-Lipschitz in context:
\begin{equation}
   \medmath{ |L(a, X) - L(a, X')| \le K \|X - X'\| ,}
\end{equation}
for some constant $K$ determined by the combined Lipschitz constants of all
components of $L$. Prediction error satisfies $\|\hat{X}_{t+1} - X_{t+1}\| \le \varepsilon$ by (A3).  
Thus, $ \medmath{|\hat{L}_{t+1}(a) - L_{t+1}(a)| \le K \varepsilon .}$
Let $\Delta_{\min}$ be the minimum loss separation between any two PQC
algorithms:
\begin{equation}
   \medmath{\Delta_{\min} = 
\min_{a \neq b} |L_{t+1}(a) - L_{t+1}(b)|.} 
\end{equation}
If $K\varepsilon < \Delta_{\min}$, then the perturbation induced by prediction
error is insufficient to change the ordering of $L_{t+1}(a)$ across all $a \in
\mathcal{A}$. Therefore, the identity of the minimizer is invariant:
\begin{equation}
    \medmath{\arg\min_a \hat{L}_{t+1}(a) 
\;=\;
\arg\min_a L_{t+1}(a).}
\end{equation}
Since APMOEA selects based on the predicted loss, it returns the same algorithm
as the oracle using the true context. Hence the decision remains stable for all
$t$ within any interval satisfying the loss-gap condition, proving the claim.
\end{proof}

\subsection{Proof of Monotonic Upgrade Security}

\begin{proof}
Let $v_t$ be the PQC version currently in use, and let $v_{t+1} > v_t$ be the
version proposed by the optimizer.  
Each transition message contains,  an authenticated version counter $\mathrm{ctr}$, a context hash $H(X_t)$, a freshness nonce $N_t$, a digital signature computed within the endpoint’s trusted hardware. We prove that no probabilistic polynomial-time (PPT) adversary can cause a
downgrade or desynchronization. A replayed message contains an old nonce $N_{t'}$ and an old version counter
$v_{t'} < v_{t+1}$.  
Freshness checks reject any message with a stale nonce.  
If the adversary rewrites $(v_{t'}, N_{t'})$ to any $(v, N)$ pair, the signature
fails verification.  
Thus replay attacks cannot induce a downgrade. To force a downgrade, the adversary must produce a message with version counter
$v < v_t$.  
Since version counters are inside the signed payload, modifying $v$ breaks the
signature, causing immediate rejection.  
No PPT adversary can forge a valid signature under the assumed security of the
PQC signature scheme. The two-phase protocol requires mutual confirmation: each endpoint transitions to
$v_{t+1}$ only after receiving a valid, freshly signed acknowledgment from the
peer. Any message dropped, delayed, or altered by the adversary results in a mismatch,
triggering rollback to $v_t$.  
Thus, neither endpoint can advance alone. A downgrade requires acceptance of some $v < v_{t+1}$.  
Both endpoints perform a monotonic version check: $\medmath{v_{\text{received}} \ge v_t \;? } $
Any $v < v_t$ is rejected regardless of message contents.  
Thus even a perfect signature forgery on an old message cannot induce downgrade. Since none of the adversarial strategies can produce a viable transcript leading
to $v < v_{t+1}$ or to inconsistent endpoint states, all accepted transitions are
strictly monotonic and synchronized.  
\end{proof}
\subsection{Proof of Latency Boundedness}
\begin{proof}
Let $T_{\mathrm{lat}}(a, X_t)$ be the end-to-end latency under context $X_t$.
Assumption (A2) states that vehicular dynamics satisfy:$\|X_{t+1} - X_t\| \le \delta.$
Each component of $T_{\mathrm{lat}}$ is continuous in the context variables
(encoding time, decoding time, network delay, and SNR-dependent effects).
Thus $T_{\mathrm{lat}}$ is itself Lipschitz: $|T_{\mathrm{lat}}(a, X_{t+1}) - T_{\mathrm{lat}}(a, X_t)|
\;\le\; C \delta,$ for some $C > 0$. Let, $\medmath{ a_t = \arg\min_{a \in \mathcal{A}} L_t(a)}$ be the algorithm selected by APMOEA which guarantees that $a_t$ lies on, or near the Pareto front of the feasible set at time $t$. Then for any $a \in \mathcal{A}$, $T_{\mathrm{lat}}(a_t, X_t) \;\le\; \max_{a' \in \mathcal{A}} T_{\mathrm{lat}}(a', X_t)$
by definition of the maximum. At time $t+1$, we have$ T_{\mathrm{lat}}(a_t, X_{t+1})
\le T_{\mathrm{lat}}(a_t, X_t) + C \delta, $ Thus,
\begin{equation}
    \medmath{T_{\mathrm{lat}}(a_t, X_t)
\;\le\;
\max_{a' \in \mathcal{A}} T_{\mathrm{lat}}(a', X_t)
+
O(\delta),}
\end{equation}
which proves that latency growth remains bounded whenever context drift is
bounded. In vehicular environments, $\delta$ corresponds to changes in mobility,
fading, and load over one optimization interval (20--50\,ms), which are small.
Hence PQC latency remains URLLC-compliant.
\end{proof}

\subsection{Robustness Under Small Prediction Error}
\begin{lemma}[\textbf{Robustness Under Small Prediction Error}]
If the prediction error satisfies the stability condition $K\varepsilon <
\Delta_{\min}$, then APMOEA selects $a_t = a^\star$.
\end{lemma}
APMOEA under admissible prediction noise matches the oracle's choice for optimal PQC, with secure migration and detection mechanisms as outlined in Algorithm~\ref{alg:transition}.

\begin{proof}
Let, $\medmath{a^\star = \arg\min_{a\in \mathcal{A}} L_{t+1}(a)}$ be the oracle-optimal PQC algorithm under the true next context $X_{t+1}$. APMOEA computes: $\medmath{a_t = \arg\min_{a} \hat{L}_{t+1}(a),}$ where $\hat{L}_{t+1}(a)$ is computed using the predicted context
$\hat{X}_{t+1}$. By (A3), $ \medmath{ \|\hat{X}_{t+1} - X_{t+1}\| \le \varepsilon.}$ By Lipschitz continuity, $\medmath{|\hat{L}_{t+1}(a) - L_{t+1}(a)| \le K \varepsilon.}$ If $K\varepsilon < \Delta_{\min}$ (the stability condition), then the
perturbation in losses is insufficient to change the identity of the minimizer.
Thus,
\begin{equation}
    \medmath{\arg\min_a \hat{L}_{t+1}(a)
\;=\;
\arg\min_a L_{t+1}(a)
\;=\;
a^\star.}
\end{equation}
Therefore, APMOEA matches the oracle’s decision whenever prediction errors are
small relative to the loss separation.  
This guarantees correct PQC selection under admissible noise.
\end{proof}
\subsection{What the tables and figures are intended to show}
Table~\ref{tab:combined_pqc} summarizes the computational, communication, and robustness characteristics of the four NIST PQC families relevant to vehicular deployment. Table~\ref{tab:rl-security} reports end-to-end latency and switching stability, highlighting the URLLC benefits of adaptive PQC selection over static baselines and the stabilizing effect of reinforcement learning. Figure~\ref{fig:three_panel_combined} evaluates robustness through (a) cost stability under bounded context manipulation, (b) comparative selection behavior across NSGA-II, RL-only, and APMOEA, and (c) runtime feasibility under hardware variability. Figure~\ref{fig:downgrade_sensitivity} confirms secure transitions by rejecting downgrade attempts and illustrates how increasing prediction error leads to higher switching frequency.

\subsection{Operational definitions used in results}

\textbf{End-to-end latency.}
Latency is defined as the sum of cryptographic processing and network delay,
$T_{\mathrm{e2e}}(a,X_t)=T_{\mathrm{crypto}}(a,\mathrm{CPU}_t)+T_{\mathrm{net}}(\mathrm{SNR}_t,\mathrm{PER}_t)$.
Cryptographic delay is obtained from NIST reference implementations and automotive CPU profiles, while network delay follows an SNR/PER-dependent NR-V2X model averaged over 200 Monte Carlo runs.
Results are evaluated against a 5--20\,ms URLLC latency budget.

\textbf{Cost under context manipulation.}
``Honest'' cost reflects nominal context evolution, while ``Attacked'' cost applies bounded perturbations to signals such as SNR and compute load to emulate sensor noise or malicious context injection.
These perturbations do not alter cryptographic primitives and only test selection robustness.

\textbf{Switching stability.}
Switching frequency measures how often the PQC configuration changes within a fixed time window (e.g., switches per 60\,s).
Lower switching indicates more stable and reliable operation.

\textbf{Downgrade resistance.}
Downgrade protection is enforced at the protocol layer via monotonic version checks, context consistency, nonce freshness, and atomic confirmation.
Transitions are accepted only if all checks pass, preventing downgrade and desynchronization attacks.

\end{document}